\documentclass[12pt,letterpaper,english]{article}
\usepackage{amsmath, amssymb, amsthm}
\usepackage[T1]{fontenc}
\usepackage[utf8]{inputenc}
\usepackage{graphicx}
\usepackage[colorlinks=true, linkcolor=black, citecolor=black, urlcolor=blue, filecolor=blue]{hyperref}
\makeatletter
\newtheorem{thm}{Theorem}
\newtheorem{prop}[thm]{Proposition}
\makeatother

\begin{document}
\title{Protected Income and Inequality Aversion\thanks{The authors warmly thank audiences at workshops and seminars at CERE-Umeå, ISER-Osaka, MDOD-Sorbonne, 
and Cal Poly San Luis Obispo for very helpful discussions. The authors also expresses gratitude to the Princeton School of Public and International Affairs for hosting Zambrano's 
sabbatical and to the Paris School of Economics for hosting Zambrano during the final stages of this research.
This research has been supported by the National Science Foundation
(USA) under Grant \# 2313969.}}
\author{Marc Fleurbaey\thanks{Paris School of Economics, Centre National de la Recherche
Scientifique, and École normale supérieure - PSL, France.} and Eduardo Zambrano\thanks{Corresponding Author. Department of Economics. Cal Poly - San Luis Obispo. USA.}}
\maketitle
\begin{abstract}
We discover a fundamental and previously unrecognized structure within the class of additively separable social welfare functions that makes it straightforward to fully characterize and elicit the social preferences of an inequality-averse evaluator. From this structure emerges a revealing question: if a large increment can be given to one individual in a society, what is the maximal sacrifice that another individual can be asked to bear for its sake? We show that the answer uncovers the evaluator’s degree of inequality aversion. In particular, all translation-invariant evaluators would sacrifice the full income of the sacrificed individual if their income were low enough and a constant amount of their income otherwise. Scale-invariant evaluators would sacrifice the full income of the sacrificed individual at all income levels if their inequality aversion was no greater than one, and a constant fraction of their income otherwise. Motivated by these findings, we propose a class of social preferences that, starting from a minimum-income level of protection, ensure a higher fraction of the sacrificed individual's income is protected the lower their income.
\end{abstract}



\section{Introduction}

Since Atkinson's (1970) and Kolm's (1976) seminal works on inequality
measurement, the notion of inequality aversion has been well understood,
and applied to many different contexts. The concept of inequality
aversion has been closely associated to the Pigou-Dalton transfer
principle, according to which social welfare does not decrease (or
inequality does not increase) when a given amount is taken from someone
and given to someone else who is less well-off, without reversing
the ranking between these two individuals.\footnote{Aboudi and Thon (2006) also study what they call the Muirhead-Dalton
principle, in which transfers are allowed to reverse the ranking between
the two individuals, but still leave the difference between them smaller
than in the initial distribution. For authoritative reviews of this
field, see Chakravarty (2009), Dutta (2002) and Cowell (2000).}

The empirical-experimental literature on people's social preferences\footnote{For a review, see Gaertner and Schokkaert (2012). }
has elicited the value of their inequality aversion about income by
making respondents choose between distributions involving different
levels of inequality and different levels of average income. For instance,
Venmans and Groom (2021) find an average coefficient of inequality
aversion of 2.9 in an empirical survey on a sample of students, while
Hurley et al. (2020) find an average coefficient of 3.27 in a representative
sample of the Ontario province of Canada. The underlying mechanism
for such ethical preference elicitation is the leaky-bucket experiment:
a distribution is less unequal than another, but has a lower average
income, reflecting the classical equality-efficiency trade-off highlighted
in Okun (1975). The meaning of the coefficient of inequality aversion
is usually presented in this way: a coefficient of 3 means that it
is worth giving a dollar to someone even if this requires taking 8
($=2^{3}$) dollars from a person who is twice as rich\textemdash thus
losing 7 dollars, i.e., 87.5\% of the amount levied from the donor.

A small literature (Fleurbaey and Michel 2001, Aboudi and Thon 2003,
Thon and Wallace 2004, Dubois 2016) has explored the relationship
between principles of transfer that allow for leakage and the value
of the coefficient of inequality aversion in a social welfare criterion
that endorses such leaky transfers.\footnote{See also Bosmans (2014) for a similar analysis applied to rank-dependent
poverty measures, involving transfers across the poverty threshold.} For instance, Fleurbaey and Michel (2001) show that if the transfers
are proportional to the initial level of the donor and recipient,
the coefficient must be at least 2, whereas if the transfers are proportional
to the final levels, the coefficient must be at least 1. 

A coefficient of 2, according to the common presentation (referring
to a gift of 1 dollar costing $2^{2}=4$) means that wasting 75\%
of the transfer when the recipient has only half the level of the
donor is acceptable, and this seems a lot. The principle of proportional
transfers (on initial levels), however, suggests something less extreme,
since when the recipient gets X, a donor who is \emph{initially} twice
as rich would only lose 2X, which means that only half of the gift
is lost in the transfer. This apparent contradiction is resolved by
noting that the reasoning in terms of one dollar gifts is only true
for infinitesimal gifts, whereas the principe of proportional transfers
applies to transfers of any size (under the proviso that they do not
reverse the order between the donor and the recipient). 

This shows that the way in which the implications of the coefficient
of inequality aversion are presented may influence the ethical choice
of the coefficient. One may find that the principle of proportional
transfers appears reasonable, since it is very common for people to
think in terms of variations in percentage, while hesitating to condone
the maxim that ``a marginal transfer from a donor to someone who
is half below can waste 75\% and still improve social welfare''\textemdash and
yet the two refer to the same coefficient of inequality aversion of
2.

In this paper we discover a fundamental and previously unrecognized structure within the class
of additively separable social welfare functions that makes it straightforward to fully characterize the social preferences of an inequality-averse evaluator.
This characterization leads directly to the following elicitation procedure: Start from a perfectly equal
distribution, and then imagine that an increment can be given to one
individual. What is the maximal sacrifice that can be imposed on another
individual for the sake of this increment? In other words, if one
individual's level of income can become infinitely high, how much
could justifiably be taken from another one in order to make this
possible? It seems to us
that this framing pushes people's intuition further toward high values
of inequality aversion. In particular, a coefficient that is below
1 would justify depriving the sacrificed individual from everything.
For a Kolm-Atkinson (scale-invariant) social welfare function, a coefficient of 2 or
more means that it is never acceptable to bring the sacrificed individual
below half of her initial level, which seems a rather generous (i.e.,
generously \emph{low}) bound\textemdash the bound is 70\% of the initial
level when the coefficient is 3. In addition to relating levels of
protection of the sacrificed individual's income to coefficients of
inequality, we actually characterize the classes of additively separable
social welfare functions that guarantee specific (absolute or relative)
levels of protection.

The paper is structured as follows. In section 2, we define the framework
and the main notions. Section 3 determines the general conditions
under which an additively separable social welfare function provides
a strictly positive protected level to the sacrificed individual.
In section 4, we characterize the class of social welfare functions
which protect a fixed fraction of the initial level, against increments
given to either one other individual, or two other individuals. Section
5 instead examines the case in which the protected level is always
at a fixed absolute difference from the initial level. Section 6 explores
how to obtain a protection that is, in proportion to the initial level,
greater when the initial level is lower, reflecting the intuitively
appealing thought that poor people deserve a relatively greater protection
than rich people. We examine in section 7 the case in which
the starting point is not perfect equality and generalize the results
to this broader class of contexts. Section 8 concludes. We also created a
companion website to the paper, available at \href{https://osf.io/tnu2q/}{https://osf.io/tnu2q/}, which contains interactive
visuals, apps and narratives that explain and motivate our results in simple and intuitive terms.

\section{Setting and main definitions}

Consider a population of $n$ individuals, each individual $i$ enjoying
income $y_{i}\geq0$, and an additively separable social welfare function
\[
W\left(y_{1},...,y_{n}\right)=\sum_{i=1}^{n}f\left(y_{i}\right),
\]
where $f$ is an increasing function with values taking an interval
in $\mathbb{R}\cup\left\{ -\infty,+\infty\right\} $. Throughout this
paper we restrict attention to this class of social welfare functions.

Income is used as the index of individual advantage in this paper,
but obviously the analysis is relevant to any setting in which individual
advantage is measured by a cardinally measurable and interpersonally
comparable index. In particular, it is possible to adjust income for
non-market aspects of quality of life that individuals enjoy or endure,
and use this ajusted income (usually called ``equivalent income'')
as the relevant index instead of ordinary income.

The social welfare function is symmetric, which reflects impartiality
among individuals. Therefore, we can focus on individual 1 (called
Ana, to fix ideas) in trade-offs with any number of fellow individuals.

Consider a starting point in which Ana and individual 2 (called Ben)
have the same income $y$. Then one can consider unequal allocations
that keep social welfare constant, i.e., spans allocations $\left(y_{1},y_{2}\right)$
such that 
\[
f\left(y_{1}\right)+f\left(y_{2}\right)=2f\left(y\right).
\]
In particular, focus on how Ana's income decreases when Ben's in raised.
One can define the function $y_{1}\left(y_{2},y\right)$ implicitly
by the equation
\[
f\left(y_{1}\left(y_{2},y\right)\right)+f\left(y_{2}\right)=2f\left(y\right),
\]
yielding the explicit form
\[
y_{1}\left(y_{2},y\right)=f^{-1}\left(2f\left(y\right)-f\left(y_{2}\right)\right).
\]
The function $y_{1}\left(y_{2},y\right)$ is increasing in $y$ and
decreasing in $y_{2}$, because $f$ is increasing and so is $f^{-1}$.

The domain of definition of $y_{1}\left(y_{2},y\right)$ depends on
$f$. For instance, if $f\left(y\right)=y$, then $y_{1}\left(y_{2},y\right)$
is defined only for $\left(y_{2},y\right)$ such that $y_{2}\leq2y$.
Let $D_{12}$ denote the domain of definition of $y_{1}\left(y_{2},y\right):$
\[
D_{12}=\left\{ \left(y_{2},y\right)|2f\left(y\right)-f\left(y_{2}\right)\geq f\left(0\right)\right\} .
\]

We can now introduce the notion of \emph{protected income}, which
is the level of income below which $y_{1}$ will never fall when $y_{2}$
increases. No matter what opportunities for high income are offered
to Ben, Ana's income will never fall below this level if social welfare
is preserved.

\textbf{Protected income:} For a given $y>0$, the level of protected
income is defined as
\[
\ddot{Y}\left(y\right)=\inf\left\{ y_{1}\left(y_{2},\hat{y}\right)|\left(y_{2},\hat{y}\right)\in D_{12},\hat{y}=y\right\} .
\]

Relatedly, we define \emph{collateral damage} as the sacrifice that
can be imposed on Ana for the sake of raising Ben's income.

\textbf{Collateral damage}: For a given $y>0$, the level of collateral
damage is defined as

\[
\ddot{L}\left(y\right)=y-\ddot{Y}\left(y\right).
\]
The \emph{relative} collateral damage is 
\[
\frac{\ddot{L}\left(y\right)}{y}=1-\frac{\ddot{Y}\left(y\right)}{y}.
\]

\section{Existence of a positive protected income}

Not all social welfare functions guarantee Ana against complete sacrifice
for the sake of Ben's stratospheric income. Actually, many of the
most commonly used social welfare functions do not protect against
complete sacrifice.
\begin{prop}
\label{existence}For all $y>0,$ one has $\ddot{Y}\left(y\right)>0$
if and only if $f$ has an upper bound strictly below $2f\left(y\right)-f\left(0\right).$
\end{prop}

\begin{proof}
Let $y$ be given. If $f$ has an upper bound strictly below $2f\left(y\right)-f\left(0\right),$
then there is $\varepsilon$ such that for all $y_{2}\geq0,$
\[
f\left(y_{2}\right)<2f\left(y\right)-f\left(0\right)-\varepsilon,
\]
implying that if $f\left(y_{1}\right)+f\left(y_{2}\right)=2f\left(y\right)>f\left(0\right)+f\left(y_{2}\right)+\varepsilon,$
then necessarily $f\left(y_{1}\right)>f\left(0\right)+\varepsilon,$
so that $y_{1}$ is bounded away from zero. In this case, $\ddot{Y}\left(y\right)>0.$

If $f$ does not have an upper bound strictly below $2f\left(y\right)-f\left(0\right),$
then for all $\varepsilon$, there is $y_{2}$ such that 
\[
f\left(y_{2}\right)>2f\left(y\right)-f\left(0\right)-\varepsilon,
\]
implying that if $f\left(y_{1}\right)+f\left(y_{2}\right)=2f\left(y\right)<f\left(0\right)+f\left(y_{2}\right)+\varepsilon,$
then necessarily $f\left(y_{1}\right)<f\left(0\right)+\varepsilon,$
so that $y_{1}$ is not bounded away from zero. In this case, $\ddot{Y}\left(y\right)=0.$ 
\end{proof}
To illustrate this result, consider the class of Kolm-Atkinson social
welfare functions, for which 
\[
f\left(y\right)=\left\{ \begin{array}{ccccccc}
\frac{y^{1-\eta}}{1-\eta} & \text{} &  &  &  & \text{for} & \eta\neq1,\eta\geq0\\
\ln y & \text{} &  &  &  & \text{for} & \eta=1
\end{array}\right.
\]
where $\eta\geq0$ is the coefficient of inequality aversion.

When $\eta\leq1$, $f$ has no upper bound and thus $\ddot{Y}\left(y\right)=0$
for all $y>0.$ When $\eta>1,$ one has $f\left(0\right)=-\infty$,
so that $2f\left(y\right)-f\left(0\right)=+\infty$, whereas $f$
is always negative. Thus, $\ddot{Y}\left(y\right)>0$ for all $y>0.$
More precisely, one has 
\[
\ddot{Y}\left(y\right)=f^{-1}\left(2f\left(y\right)\right)=2^{\frac{1}{1-\eta}}y.
\]

This is a quite significant result because values of $\eta=0.5$ and
$1$ are frequent in the applied literature on inequalities. With
such values, the ultimate sacrifice of giving all of one's income
for the sake of a billionaire is considered acceptable, and this appears
hard to defend. If one would like to avoid such a repugnant conclusion,
it is imperative to pick $\eta>1$.

Another important class is formed by the Kolm-Pollak social welfare
functions, for which 
\[
f\left(y\right)=\left\{ \begin{array}{ccccccc}
-e^{-\alpha y} & \text{} &  &  &  & \text{for} & \alpha>0\\
y & \text{} &  &  &  & \text{for} & \alpha=0
\end{array}\right.
\]
where $\alpha\geq0$ is the coefficient of inequality aversion. When
$\alpha=0,$ one obviously has $\ddot{Y}\left(y\right)=0$ for all
$y>0.$ When $\alpha>0$, $f$ is negative with upper bound zero,
whereas $2f\left(y\right)-f\left(0\right)$ is non-positive when $y\leq\frac{\ln2}{\alpha}$,
and thus $\ddot{Y}\left(y\right)=0$ for all $y\leq\frac{\ln2}{\alpha}.$
When $y>\frac{\ln2}{\alpha}$, one has 
\[
\ddot{Y}\left(y\right)=f^{-1}\left(2f\left(y\right)\right)=y-\frac{\ln2}{\alpha}.
\]

Unlike the Kolm-Atkinson class, members of the Kolm-Pollak class never
guarantee a positive protected level of income for all values of $y,$
no matter how high the coefficient of inequality aversion is. This
can be seen as a serious disadvantage for this class, although the
level of $y$ below which the protected income falls to zero becomes
vanishingly small when $\alpha$ is high.

A graphical representation of $\ddot{Y}\left(y\right)$ is easy to
build upon the graph of $f$. Given that $f$ must have an upper bound,
there is no loss of generality to let $f\left(+\infty\right)=0.$
One then has $\ddot{Y}\left(y\right)=f^{-1}\left(2f\left(y\right)\right).$
This is illustrated in Fig. 1.

\begin{figure}
\begin{centering}
\includegraphics[scale=0.5]{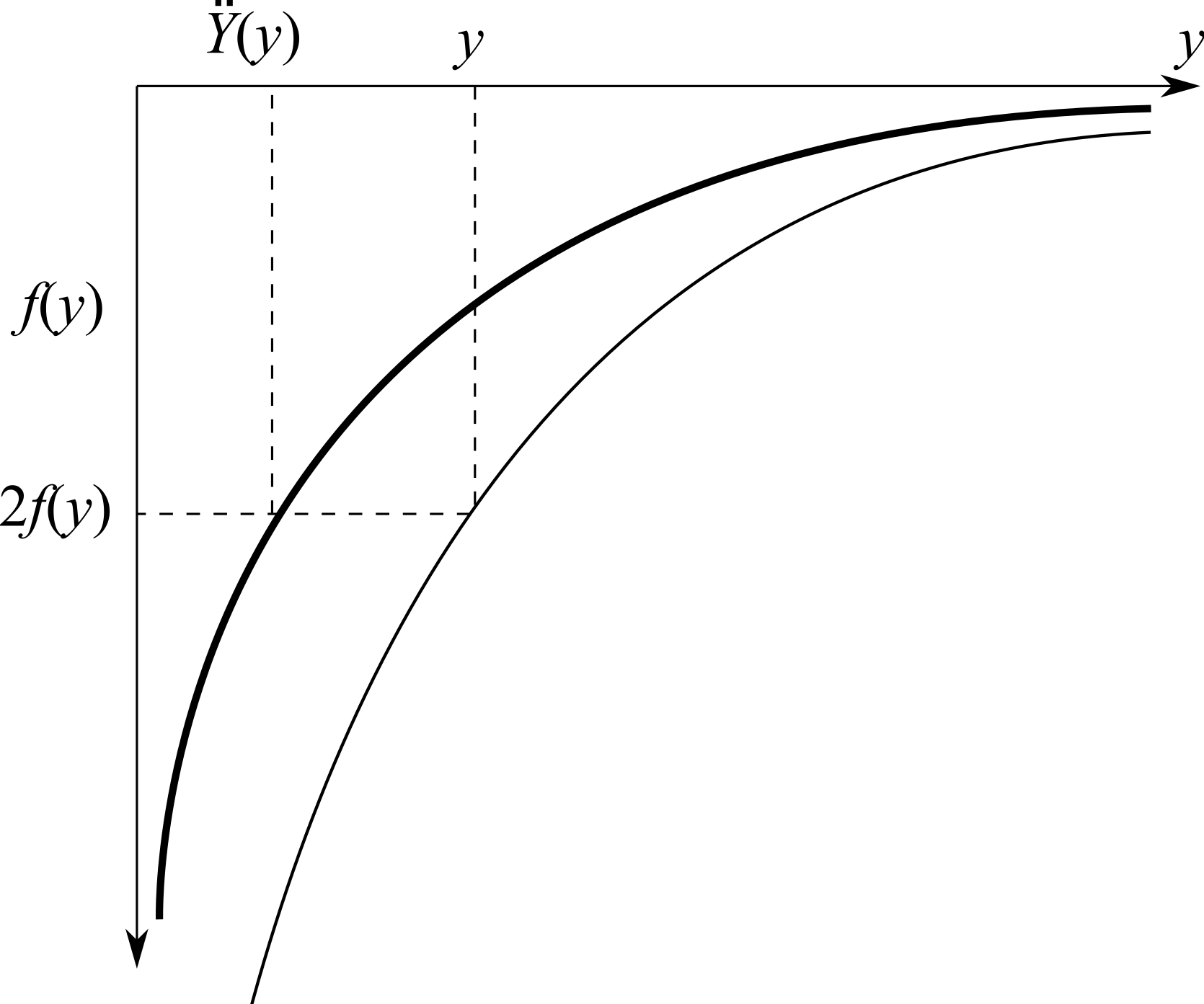}
\par\end{centering}
\caption{Representation of $\ddot{Y}\left(y\right)$}

\end{figure}

\section{Protecting a fraction of income}

Consider again the Kolm-Atkinson social welfare functions, for the
relevant case $\eta>1.$ Since
\[
\ddot{Y}\left(y\right)=2^{\frac{1}{1-\eta}}y,
\]
the protected income is always a fixed fraction of the initial (equal)
income. The fraction $2^{\frac{1}{1-\eta}}$ tends to zero when $\eta$
is close to 1, and tends to one when $\eta$ is very high. So, this
fraction spans the whole range, which connects the coefficient of
inequality aversion to the protected income in a convenient way.

The next proposition characterizes the class of social welfare functions
protecting a fixed fraction of income.
\begin{prop}
\label{prop:fixed prop12}The subclass of social welfare functions
satisfying, for all $y>0$ and a fixed $\lambda\in\left(0,1\right)$,
the property $\ddot{Y}\left(y\right)=\lambda y$ is generated by the
function 
\[
f\left(y\right)=-e^{g\left(\ln y\right)},
\]
where $g$ is defined as follows. It is an arbitrary continuous decreasing
function over the interval $\left[0,-\ln\lambda\right]$ such that
$g\left(\ln\lambda\right)-g\left(0\right)=\ln2$, and over the rest
of $\mathbb{R}$ it is defined in a periodic way, i.e., $g\left(\left(k+1\right)\ln\lambda+a\right)-g\left(k\ln\lambda+a\right)=\ln2$
for all $k\in\mathbb{Z}_{+},$$a\in\left[0,-\ln\lambda\right]$.
\end{prop}

\begin{proof}
From Prop. \ref{existence}, $f$ must have an upper bound. There
is no loss of generality in positing $f\left(+\infty\right)=0.$ (Note
that $f$ is thus negative.) One then has $\ddot{Y}\left(y\right)=f^{-1}\left(2f\left(y\right)\right).$

We are looking for the class of $f$ satisfying, for all $y>0$: 
\[
f^{-1}\left(2f\left(y\right)\right)=\lambda y
\]
for a fixed $\lambda.$ This can be written as 
\[
2f\left(y\right)=f\left(\lambda y\right).
\]
Letting $g=\ln\left(-f\circ e\right)$, a change of variables $x=\ln y$
(note that $g$ is decreasing and is defined on $\mathbb{R}$) yields
\[
\ln2+g\left(x\right)=g\left(\ln\lambda+x\right).
\]
This property is satisfied by any function $g$ such that for all
$k\in\mathbb{Z}_{+}$,
\[
g\left(k\ln\lambda\right)=g\left(0\right)+k\ln2
\]
(this defines a line of points of intercept $g\left(0\right)$ and
negative slope $\ln2/\ln\lambda$) and $g$ is, suitably detrended,
periodic across intervals $\left[\left(k+1\right)\ln\lambda,k\ln\lambda\right]$,
i.e., for $a\in\left[0,-\ln\lambda\right]$, 
\[
g\left(k\ln\lambda+a\right)+\ln2=g\left(\left(k+1\right)\ln\lambda+a\right).
\]

Once such a function has been picked, one then retrieves 
\[
f\left(y\right)=-e^{g\left(\ln y\right)}.
\]
\end{proof}
Since the graph of $g$ contains a line of points, the case of a linear
$g$ is obviously salient. One then has $g\left(x\right)=x\ln2/\ln\lambda$,
and
\[
f\left(y\right)=-y^{\ln2/\ln\lambda},
\]
which brings us back to the Kolm-Atkinson class, for $\eta=1-\ln2/\ln\lambda,$
and this implies that the fraction of protected income is indeed $2^{\frac{1}{1-\eta}}.$ 

Let us now extend the notion of protected income to the case in which
individuals 2 and 3 see their incomes rise without limit. Let the
function $y_{1}\left(y_{2},y_{3},y\right)$ be defined implicitly
by the equation
\[
f\left(y_{1}\left(y_{2},y\right)\right)+f\left(y_{2}\right)+f\left(y_{3}\right)=3f\left(y\right),
\]
and explicitly as
\[
y_{1}\left(y_{2},y_{3},y\right)=f^{-1}\left(3f\left(y\right)-f\left(y_{2}\right)-f\left(y_{3}\right)\right).
\]
It is increasing in $y$ and decreasing in both $y_{2},y_{3}.$

The domain of definition of $y_{1}\left(y_{2},y_{3},y\right)$, denoted
$D_{123}$ is defined as follows
\[
D_{12}=\left\{ \left(y_{2},y_{3},y\right)|3f\left(y\right)-f\left(y_{2}\right)-f\left(y_{3}\right)\geq f\left(0\right)\right\} .
\]

\textbf{Protected income against two:} For a given $y>0$, the level
of protected income against two other individuals is defined as
\[
\dddot{Y}\left(y\right)=\inf\left\{ y_{1}\left(y_{2},y_{3},\hat{y}\right)|\left(y_{2},y_{3},\hat{y}\right)\in D_{123},\hat{y}=y\right\} .
\]

In the case of the Kolm-Atkinson class, one obtains
\[
\dddot{Y}\left(y\right)=3^{\frac{1}{1-\eta}}y,
\]
so that once again, a fixed fraction of $y$ is guaranteed to Ana.
If one combines the requirement to protect a fixed (different) fraction
of income against one and two other individuals, one obtains a characterization
of this class (up to an affine transform).
\begin{prop}
\label{prop:fixedprop123}The Kolm-Atkinson subclass of social welfare
functions is (up to an affine transform) characterized by the requirement
that there exist $\lambda,\mu\in\left(0,1\right)$ such that for all
$y>0$, $\ddot{Y}\left(y\right)=\lambda y$ and $\dddot{Y}\left(y\right)=\mu y.$
One must then have $\lambda=2^{\frac{1}{1-\eta}}$ and $\mu=3^{\frac{1}{1-\eta}}$
for some $\eta>1$.
\end{prop}

\begin{proof}
From Prop. \ref{prop:fixed prop12} we know that $f\left(y\right)=-e^{g\left(\ln y\right)}$
for a function $g$ such that 
\[
g\left(\left(k+1\right)\ln\lambda+a\right)-g\left(k\ln\lambda+a\right)=\ln2
\]
for all $k\in\mathbb{Z}_{+},$$a\in\left[0,-\ln\lambda\right]$. By
a similar reasoning, the function $g$ must also satisfy
\[
g\left(\left(k+1\right)\ln\mu+a\right)-g\left(k\ln\mu+a\right)=\ln3
\]
for all $k\in\mathbb{Z}_{+},$$a\in\left[0,-\ln\mu\right]$. In particular,
for all $k\in\mathbb{Z}_{+},$
\[
g\left(k\ln\lambda\right)-g\left(0\right)=k\ln2
\]
\[
g\left(k\ln\mu\right)-g\left(0\right)=k\ln3.
\]
Given the quasi-periodic behavior of $g$, there exist $b<0<c$ such
that for all $x\in\mathbb{R}$, 
\[
b<g\left(x\right)-g\left(0\right)-x\frac{\ln2}{\ln\lambda}<c
\]
\[
b<g\left(x\right)-g\left(0\right)-x\frac{\ln3}{\ln\mu}<c.
\]
This is possible only if 
\[
\frac{\ln2}{\ln\lambda}=\frac{\ln3}{\ln\mu}.
\]
Let $1-\eta$ be this common (negative) ratio. This proves that $\lambda=2^{\frac{1}{1-\eta}}$
and $\mu=3^{\frac{1}{1-\eta}}$ for some $\eta>1$.

Without loss of generality (i.e., up to an affine transform of $W$),
we can let $\ensuremath{f}\left(+\infty\right)=0$ and $\ensuremath{f}\left(1\right)=\frac{1}{1-\text{\ensuremath{\eta}}}$.
Then

\[
\left\{ \begin{array}{c}
f\text{\ensuremath{\left(2^{\ensuremath{\frac{1}{1-\eta}}}y\right)}}=2\ensuremath{f}\left(y\right)\\
f\text{\ensuremath{\left(3^{\ensuremath{\frac{1}{1-\eta}}}y\right)}}=3\ensuremath{f}\left(y\right)
\end{array}\right.
\]

The rest of the argument follows the proof of the first half of Theorem
A in Tossavainen and Haukkanen (2007). Let $S=\left\{ \left.2^{\ensuremath{\frac{m}{1-\eta}}}3^{\ensuremath{\frac{n}{1-\eta}}}\right|m,n\in\mathbb{Z}\right\} $
and pick $y$$\in S.$ Then, 

\[
f\left(y\right)=f\left(2^{\ensuremath{\frac{m}{1-\eta}}}3^{\ensuremath{\frac{n}{1-\eta}}}\right)=2^{m}f\left(3^{\ensuremath{\frac{n}{1-\eta}}}\right)=2^{m}3^{n}f\left(1\right)
\]
\[
=\left(2^{\ensuremath{\frac{m}{1-\eta}}}3^{\ensuremath{\frac{n}{1-\eta}}}\right)^{1-\eta}f\left(1\right)=f\left(1\right)y^{1-\eta}=\frac{y^{1-\eta}}{1-\eta}.
\]

We have established that $f\left(y\right)=\frac{y^{1-\eta}}{1-\eta}$
for $y\in S.$ In order to extend the result to all $y$ in the interval
$\left(0,+\infty\right)$ we need to show that $S$ is dense in $\left(0,+\infty\right).$
Pick $y$ in $\left(0,+\infty\right)$ and let $T=\left\{ \left.m+n\log_{2}3\right|m,n\in\mathbb{Z}\right\} $.
Since $\log_{2}3$ is irrational this means that $T$ is dense in
$\mathbb{R}$. Since $\left(1-\eta\right)\log_{2}y\in\mathbb{R}$,
there is a sequence $\left\{ t_{j}\right\} $$\subset T$ such that
\[
\lim_{j\rightarrow\infty}t_{j}=\left(1-\eta\right)\log_{2}y,
\]
which implies that $\lim_{j\rightarrow\infty}\frac{t_{j}}{1-\eta}=\log_{2}y,$
or $\lim_{j\rightarrow\infty}2^{\frac{t_{j}}{1-\eta}}=y.$ The next
step is to show that $\left\{ 2^{\frac{t_{j}}{1-\eta}}\right\} \subset S.$
This is readily seen, since $2^{\frac{t_{j}}{1-\eta}}=2^{\frac{m+n\log_{2}3}{1-\eta}}=2^{\ensuremath{\frac{m}{1-\eta}}}\left(2^{\log_{2}3}\right)^{\ensuremath{\frac{n}{1-\eta}}}=2^{\ensuremath{\frac{m}{1-\eta}}}3^{\ensuremath{\frac{n}{1-\eta}}}$
for some $m,n\in\mathbb{Z}$, which means that $2^{\frac{t_{j}}{1-\eta}}\in S$
for each $j$. It follows that $S$ is dense in $\left(0,+\infty\right)$
and therefore that $\ensuremath{f}\left(y\right)=\frac{y^{1-\eta}}{1-\eta}$
for all $y\in\left(0,+\infty\right)$, which is what we wanted to
show.
\end{proof}

\section{Constant Collateral Damage}

We have criticized the Kolm-Pollak social welfare functions for not
guaranteeing a positive income when the initial income is low. But
it is worth revisiting the property of a constant collateral damage,
i.e., a protected income that is always at the same distance from
the initial income, provided the latter is greater than this constant
collateral damage.

Here again, we assume that $f$ has an upper bound and that $f\left(+\infty\right)=0.$
The next proposition characterizes the class of social welfare functions
exhibiting constant collateral damage.
\begin{prop}
\label{prop:constcodam21}The subclass of social welfare functions
satisfying, for a fixed $\varDelta\in\left(0,+\infty\right)$ and
all $y>\varDelta$, the property $\ddot{Y}\left(y\right)=y-\varDelta$
is generated by the function 
\[
f\left(y\right)=-e^{g\left(y\right)},
\]
where $g$ is defined as follows. It is an arbitrary continuous decreasing
function over the interval $\left[0,\varDelta\right]$ such that $g\left(\varDelta\right)-g\left(0\right)=-\ln2$,
and over the rest of $\left(0,+\infty\right)$ it is defined in a
periodic way, i.e., $g\left(\left(k+1\right)\varDelta+a\right)-g\left(k\varDelta+a\right)=-\ln2$
for all $k\in\mathbb{Z}_{+},$ $a\in\left[0,\varDelta\right]$.
\end{prop}

\begin{proof}
Constant collateral damage means that for all $y>\varDelta,$
\[
f\left(y-\varDelta\right)=2f\left(y\right).
\]
Let $g=\ln f.$ For all $y>0,$ one has 
\[
g\left(y+\varDelta\right)=g\left(y\right)-\ln2.
\]
For $y=0,$ one obtains $g\left(\varDelta\right)-g\left(0\right)=-\ln2,$
and for $y=k\varDelta+a,$ one obtains $g\left(\left(k+1\right)\varDelta+a\right)-g\left(k\varDelta+a\right)=-\ln2.$ 
\end{proof}
In similar fashion as in the previous section, one can add the requirement
that constant collateral damage also occurs against two other individuals,
and derive a characterization of the Kolm-Pollak class of social welfare
functions.
\begin{prop}
\label{prop:The-Kolm-Pollak-subclass}The Kolm-Pollak subclass of
social welfare functions is (up to an affine transform and above a
minimum level of income) characterized by the requirement that there
exist $\varOmega>\varDelta>0$ such that for all $y>\varOmega$, $\ddot{Y}\left(y\right)=y-\varDelta$
and $\dddot{Y}\left(y\right)=y-\varOmega.$ One must then have $\varDelta=\frac{\ln2}{\alpha}$
and $\varOmega=\frac{\ln3}{\alpha}$ for some $\alpha>0$.
\end{prop}

\begin{proof}
See Appendix \ref{subsec:Proof-of-Proposition-5}.
\end{proof}

\section{Progressive Protection}

The proportional protected income, with $\eta>1,$ is an improvement
over the constant collateral damage in the sense that in the latter
case, collateral damage, in proportion to income, is greater at lower
levels of income. However, the fact that \emph{relative} collateral
damage is indeed constant with the Kolm-Atkinson approach, no matter
what level of income one is considering, may be viewed as problematic,
too. One might want to consider orderings that protect a higher fraction
of Ana's income (and therefore allow lower relative collateral damage)
the lower her income is. We turn our attention now to a novel class
of orderings where this would be so. 

Fix a parameter $c>0.$ Let $f$ be defined, for $y\geq c,$ by:
\[
f\left(y\right)=\left\{ \begin{array}{ccccccc}
\frac{\left(\ln\frac{y}{c}\right)^{1-\gamma}}{1-\gamma} & \text{} &  &  &  & \text{for} & \gamma\neq1,\gamma\geq0\\
\ln\ln\frac{y}{c} & \text{} &  &  &  & \text{for} & \gamma=1.
\end{array}\right.
\]

We first note that, once again, $\gamma\leq1$ does not embed enough
inequality aversion in that, for those social welfare functions, there
is no substantial protection. This follows directly from the fact
that $f$ has no upper bound in this case, so that $2f\left(y\right)-f\left(y_{1}\right)$
can be passed by $f\left(y_{2}\right)$ for $y_{2}$ great enough,
no matter how low $y_{1}>c$ is.

On the other hand, for $\gamma>1$, these social welfare functions
exhibit the feature that a higher fraction of income is protected
the lower Ana's income is and therefore the orderings have increasing
relative collateral damage.
\begin{prop}
Consider the subclass of social welfare functions generated by $f$
as defined above. Then, for all $y>c$,\textup{$\ddot{Y}\left(y\right)=y^{\ddot{\beta}}c^{1-\ddot{\beta}},$}
where \textup{$\ddot{\beta}=2^{\ensuremath{\frac{1}{1-\gamma}}}$
.}
\end{prop}

\begin{proof}
We compute that

\[
y_{1}\left(y_{2};y\right)=c\cdot e^{\left[2\left(\ln\frac{y}{c}\right)^{1-\gamma}-\left(\ln\frac{y_{2}}{c}\right)^{1-\gamma}\right]^{\frac{1}{1-\gamma}}}
\]

Notice that $\lim_{y_{2}\rightarrow+\infty}\left[2\left(\ln\frac{y}{c}\right)^{1-\gamma}-\left(\ln\frac{y_{2}}{c}\right)^{1-\gamma}\right]^{\frac{1}{1-\gamma}}=2^{\ensuremath{\frac{1}{1-\gamma}}}\ln\frac{y}{c},$
so that 
\[
\ddot{Y}\left(y\right)=c\cdot e^{\ln\frac{y}{c}\cdot2^{\ensuremath{\frac{1}{1-\gamma}}}}=c\cdot\left(\frac{y}{c}\right)^{2^{\ensuremath{\frac{1}{1-\gamma}}}}=y^{\ddot{\beta}}c^{1-\ddot{\beta}}.
\]
\end{proof}
The coefficient $\ddot{\beta}$ is the \emph{income elasticity of
protected income}. For $\gamma>1$ the elasticity $\ddot{\beta}$
is thus between zero and one. This implies that $\frac{\ddot{Y}\left(y\right)}{y}=\left(\frac{c}{y}\right)^{1-\ddot{\beta}}$
is a decreasing function of equivalent income, and consequently the
relative collateral damage $1-\left(\frac{c}{y}\right)^{1-\ddot{\beta}}$
is an increasing function of equivalent income, as desired.

Note that $\ddot{\beta}$ is increasing in $\gamma$, that $\lim_{\gamma\rightarrow\infty}2^{\ensuremath{\frac{1}{1-\gamma}}}=1$,
and that $\lim_{\gamma\downarrow1}2^{\ensuremath{\frac{1}{1-\gamma}}}=0$.
In words: a large $\gamma$ corresponds to the situation where most
of Ana's equivalent income is protected as Ben's income grows without
bound whereas a small $\gamma$ corresponds to a situation where only
a baseline level $c$ of Ana's income is protected as Ben's income
grows without bound. 

For these orderings, protected income is therefore a weighted geometric
average of a normatively determined baseline income level, $c$, and
the original income level, $y$, and where the weight on income is
precisely the income elasticity $\ddot{\beta}$. Because this income
elasticity is the same for high or low income levels, we can identify
these orderings as those having constant protected income elasticities
and, when staying within this class, knowing the level $c>0$ and
the magnitude of the elasticity $\ddot{\beta}$ is sufficient to pin
down the entire ordering, as one can recover the parameter $\gamma$
through the expression $\ddot{\beta}=2^{\ensuremath{\frac{1}{1-\gamma}}}$. 

We show that constant protected income elasticities and a baseline
protected income level indeed characterize this class of social welfare
functions.
\begin{prop}
\textup{\label{CPIE}}The subclass of social welfare functions generated
by $f$(as defined in this section) is, up to an affine transform,
characterized by the following property: There exist \textup{$\ddot{\beta},\dddot{\beta}\in(0,1),$
$c>0$ such that, for all $y\in\left(c,\infty\right)$,
\[
\ddot{Y}\left(y\right)=y^{\ddot{\beta}}c^{1-\ddot{\beta}}\text{ and \ensuremath{\dddot{Y}\left(y\right)=y^{\dddot{\beta}}c^{1-\dddot{\beta}}.}}
\]
In particular, necessarily there is $\gamma>1$ such that }$\ddot{\beta}=2^{\ensuremath{\frac{1}{1-\gamma}}},\dddot{\beta}=3^{\ensuremath{\frac{1}{1-\gamma}}}$.
\end{prop}

\begin{proof}
See Appendix \ref{subsec:Proof-of-Proposition-seven}.
\end{proof}

\section{Unequal status quo }

So far, we have taken equality as a starting point and examined what
sacrifice could be imposed on Ana for the sake of one or two other
individuals. It is interesting also to generalize the analysis and
examine starting points containing some inequality.

Consider any starting point $\left(y_{1},y_{2}\right)$. How much
of Ana's income is protected against a rise in Ben's income? This
can be derived in straightforward manner from the previous notions.
First, let us define the equally-distributed equivalent, relative
to function $f$, of the unequal distribution $\left(y_{1},y_{2}\right):$
\[
EE_{f}\left(y_{1},y_{2}\right):=f^{-1}\left[\frac{1}{2}f\left(y_{1}\right)+\frac{1}{2}f\left(y_{2}\right)\right].
\]

Retaining the same notation $\ddot{Y}$ for both an unequal and an
equal status quo, it is clear that
\[
\ddot{Y}\left(y_{1},y_{2}\right)=\ddot{Y}\left(EE_{f}\left(y_{1},y_{2}\right)\right),
\]
since an egalitarian status quo that would yield the same social welfare
as $\left(y_{1},y_{2}\right)$ is suitable as a benchmark for the
social welfare level to keep constant while Ana's income is decreased
for the sake of Ben's. Retaining the convention that $f\left(+\infty\right)=0$
for a bounded function, one also computes
\[
\ddot{Y}\left(y_{1},y_{2}\right)=f^{-1}\left(f\left(y_{1}\right)+f\left(y_{2}\right)\right).
\]

It is worth observing that there is a correspondence between the properties
of protected income (and collateral damage) and invariance properties
of $EE_{f}$ which also characterize the corresponding social welfare
functions. Let $EE_{\eta}\left(y_{1},y_{2}\right)$ denote the equally-distributed
equivalent for a Kolm-Atkinson social preference with inequality aversion
coefficient $\eta>0$. Let $EE_{\alpha}\left(y_{1},y_{2}\right)$
denote the equally-distributed equivalent for a Kolm-Pollak social
preference with inequality aversion coefficient $\alpha>0$. Last,
given $c>0$, $y_{1}\geq c$ and $y_{2}\geq c$, let $EE_{\gamma}\left(y_{1},y_{2},c\right)$
denote the equally-distributed equivalent for the social preference
defined in Section 6 with inequality aversion coefficient $\gamma>0$.
We obtain the following:
\begin{prop}
For the Kolm-Atkinson class, $EE_{\eta}\left(y_{1},y_{2}\right)$
is scale invariant: for every \textup{$\kappa>0$, 
\[
EE_{\eta}\left(\kappa y_{1},\kappa y_{2}\right)=\kappa EE_{\eta}\left(y_{1},y_{2}\right).
\]
}For the Kolm-Pollack class, $EE_{\alpha}\left(y_{1},y_{2}\right)$
is translation invariant: for every $k>0,$
\[
EE_{\alpha}\left(y_{1}+k,y_{2}+k\right)=EE_{\alpha}\left(y_{1},y_{2}\right)+k.
\]
 For the class introduced in Section 6, given $c>0$, $y_{1}\geq c$\textup{
and $y_{2}\geq c$,} $EE_{\gamma}\left(y_{1},y_{2},c\right)$ is compound
invariant: for every $\rho>0$,
\[
EE_{\gamma}\left(y_{1}^{\rho},y_{2}^{\rho},c^{\rho}\right)=EE_{\gamma}\left(y_{1},y_{2},c\right)^{\rho}.
\]
\end{prop}

\begin{proof}
This is known for the first two classes, and we only prove the third
case. 

Let $\gamma\neq1.$ Then $EE=EE_{\gamma}\left(y_{1},y_{2},c\right)$
satisfies
\[
\frac{1}{2}\frac{\left[\ln\left(\frac{y_{1}}{c}\right)\right]{}^{1-\gamma}}{1-\gamma}+\frac{1}{2}\frac{\left[\ln\left(\frac{y_{2}}{c}\right)\right]{}^{1-\gamma}}{1-\gamma}=\frac{\left[\ln\left(\frac{EE}{c}\right)\right]{}^{1-\gamma}}{1-\gamma}.
\]
Let $\rho>0$. Then $EE^{*}=EE_{\gamma}\left(y_{1}^{\rho},y_{2}^{\rho},c^{\rho}\right)$satisfies
\[
\frac{1}{2}\frac{\left[\ln\left(\frac{y_{1}}{c}\right)^{\rho}\right]{}^{1-\gamma}}{1-\gamma}+\frac{1}{2}\frac{\left[\ln\left(\frac{y_{2}}{c}\right)^{\rho}\right]{}^{1-\gamma}}{1-\gamma}=\frac{\left[\ln\left(\frac{EE^{*}}{c^{\rho}}\right)\right]^{1-\gamma}}{1-\gamma},
\]
that is,
\[
\frac{1}{2}\frac{\left[\rho\ln\left(\frac{y_{1}}{c}\right)\right]^{1-\gamma}}{1-\gamma}+\frac{1}{2}\frac{\left[\rho\ln\left(\frac{y_{2}}{c}\right)\right]^{1-\gamma}}{1-\gamma}=\frac{\left[\ln\left(\frac{EE^{*}}{c^{\rho}}\right)\right]^{1-\gamma}}{1-\gamma},
\]

\[
\frac{1}{2}\frac{\rho^{1-\gamma}\left[\ln\left(\frac{y_{1}}{c}\right)\right]^{1-\gamma}}{1-\gamma}+\frac{1}{2}\frac{\rho^{1-\gamma}\left[\ln\left(\frac{y_{2}}{c}\right)\right]^{1-\gamma}}{1-\gamma}=\frac{\left[\ln\left(\frac{EE^{*}}{c^{\rho}}\right)\right]^{1-\gamma}}{1-\gamma}.
\]

\[
\rho^{1-\gamma}\frac{\left[\ln\left(\frac{EE}{c}\right)\right]{}^{1-\gamma}}{1-\gamma}=\frac{\left[\ln\left(\frac{EE^{*}}{c^{\rho}}\right)\right]^{1-\gamma}}{1-\gamma}
\]

\[
\frac{\ln\left(\left(\frac{EE}{c}\right)^{\rho}\right)^{1-\gamma}}{1-\gamma}=\frac{\left[\ln\left(\frac{EE^{*}}{c^{\rho}}\right)\right]^{1-\gamma}}{1-\gamma}
\]
 from where it follows that $EE^{\rho}=EE^{*}$. 

Now let $\gamma=1.$ Then $EE=EE_{1}\left(y_{1},y_{2},c\right)$ satisfies
\[
\frac{1}{2}\ln\ln\frac{y_{1}}{c}+\frac{1}{2}\ln\ln\frac{y_{2}}{c}=\ln\ln\left(\frac{EE}{c}\right).
\]
Fix $\rho>0$. Then $EE^{*}=EE_{1}\left(y_{1}^{\rho},y_{2}^{\rho},c^{\rho}\right)$satisfies
\[
\frac{1}{2}\ln\ln\left(\frac{y_{1}}{c}\right)^{\rho}+\frac{1}{2}\ln\ln\left(\frac{y_{2}}{c}\right)^{\rho}=\ln\ln\left(\frac{EE^{*}}{c^{\rho}}\right),
\]
that is,
\[
\frac{1}{2}\ln\left(\rho\ln\frac{y_{1}}{c}\right)+\frac{1}{2}\ln\left(\rho\ln\frac{y_{2}}{c}\right)=\ln\ln\left(\frac{EE^{*}}{c^{\rho}}\right),
\]

\[
\frac{1}{2}\left[\ln\rho+\ln\ln\frac{y_{1}}{c}\right]+\frac{1}{2}\left[\ln\rho+\ln\ln\frac{y_{2}}{c}\right]=\ln\ln\left(\frac{EE^{*}}{c^{\rho}}\right),
\]

\[
\ln\rho+\frac{1}{2}\ln\ln\frac{y_{1}}{c}+\frac{1}{2}\ln\ln\frac{y_{2}}{c}=\ln\ln\left(\frac{EE^{*}}{c^{\rho}}\right),
\]

\[
\ln\rho+\ln\ln\left(\frac{EE}{c}\right)=\ln\ln\left(\frac{EE^{*}}{c^{\rho}}\right),
\]

\[
\ln\left(\rho\ln\left(\frac{EE}{c}\right)\right)=\ln\ln\left(\frac{EE^{*}}{c^{\rho}}\right),
\]

\[
\ln\ln\left(\frac{EE}{c}\right)^{\rho}=\ln\ln\left(\frac{EE^{*}}{c^{\rho}}\right),
\]
from where it follows that $EE^{\rho}=EE^{*}$, which is what we wanted
to show.
\end{proof}

\section{Conclusions}

This paper reveals a fundamental structure within additively separable social welfare functions that has gone unrecognized until now. Through this discovery, we provide novel characterizations of well-known social preferences using intuitively interpretable parameters. The practical significance of these parametric representations lies in their straightforward elicitation and validation, making them particularly valuable for empirical applications. To wit: the Kolm-Pollak social preferences can be
described by noting that they exhibit \emph{constant difference protected
income} (CDPI), and we write $f\left(y\right)=-2^{-\frac{y}{L}}$
where $L>0$ identifies the largest collateral damage that is acceptable
to the evaluator in a two-person evaluation for all income levels,
and therefore $\ddot{Y}\left(y\right)=y-L$ for all $y>0$. In turn,
the Kolm-Atkinson social preferences that display positive protection
can be described by noting that they exhibit \emph{constant relative
protected income} (CRPI), and we write $f\left(y\right)=-y^{\frac{1}{\log_{2}\lambda}}$
where $\lambda\in\left(0,1\right)$ identifies the smallest fraction
of income that will be protected by the evaluator in a two-person
evaluation for all income levels, and therefore $\ddot{Y}\left(y\right)=\lambda y$
for all $y>0$. Last, we can say that the sub-class of social preferences
we proposed in Section 6 that displays positive protection exhibits
\emph{constant protected income elasticities} (CPIE), and we write
$f\left(y\right)=-\left(\ln\frac{y}{c}\right)^{\frac{1}{\log_{2}\beta}}$
for $y>c$, where $\beta\in\left(0,1\right)$ is the income elasticity
of protected income for all income levels in a two-person evaluation,
$c$ is a normatively determined lower bound on protected income for
all $\beta\in\left(0,1\right),$ and $\ddot{Y}\left(y\right)=y^{\beta}c^{1-\beta}$
for all $y>c$. 

With the contribution of this paper, the menu of choices for how to
elicit ethical preferences about inequality aversion is enriched with
a framing that is quite different from the leaky-bucket experiment
usually employed, explicitly or implicitly, in surveys and experiments.
Giving sacrificed individuals a protection against the collateral
damage of pursuing the improvement in the lot of other people seems
an attractive principle, and it can serve to calibrate the coefficient
of inequality aversion. 

Moreover, we conjecture that this thought experiment tends to justify
higher values of inequality aversion than the leaky-bucket experiment.
Two cognitive mechanisms may induce this phenomenon. First, the leaky-bucket
thought experiment involves a loss to society, as part of the transfer
is wasted, and loss aversion is likely to kick in and generate a jaundiced
view of the transfer. Second, the usual application involving a discrepancy
between finite transfers and a numerical calibration of the coefficient
that is only justified for infinitesimal transfers biases the analysis
in the direction of a smaller coefficient (e.g., the proportional
transfer principle requires a coefficient of 2 but only a coefficient
of 1 for infinitesimal transfers, and it is wrong to infer from the
respondents' acceptance of proportional transfers that their coefficient
is 1).

In the thought experiment proposed here, loss aversion pushes in the
direction of greater inequality aversion, since the loss that is imposed
on the sacrificed individual looms large in the ethical evaluation.
We therefore propose that empirical estimations of inequality aversion
in the population should no longer be confined to leaky-bucket types
of experiments (and certainly not confuse infinitesimal and large
transfers) and should instead, or at least additionally, involve a
``protected level'' type of experiment in order to obtain another
relevant estimate.

Having representations of social preferences that
are easy to motivate and interpret, as we have aimed to provide in
this paper and its companion website, \href{https://osf.io/tnu2q/}{https://osf.io/tnu2q/}, 
can lower the barriers to adoption of these social welfare
tools by applied economists. This would be a welcome development to
anyone interested in fostering the inequality sensitive evaluation
of economic and social policies.

\section{References}

Aboudi, R., \& Thon, D. (2003). Transfer principles and relative inequality
aversion a majorization approach. Mathematical social sciences, 45(3),
299-311.

Aboudi, R., \& Thon, D. (2006). Refinements of Muirhead's Lemma and
income inequality. Mathematical social sciences, 51(2), 201-216.

Atkinson, A. B. (1970), \textquotedblleft On the Measurement of Inequality,\textquotedblright{}
Journal of Economic Theory, 2, 244\textendash 263.

Bosmans, K. (2014). Distribution-sensitivity of rank-dependent poverty
measures. Journal of Mathematical Economics, 51, 69-76.

Chakravarty, S. R. (2009). Inequality, polarization and poverty. Advances
in distributional analysis. New York.

Cowell, F. A. (2000). Measurement of inequality. \emph{Handbook of
income distribution}, vol. 1, 87\textendash 166.

Dubois, M. (2016). A note on the normative content of the Atkinson
inequality aversion parameter. Economics Bulletin.

Dutta, B. (2002). Inequality, poverty and welfare. Handbook of social
choice and welfare, 1, 597-633.

Fleurbaey M., P. Michel (2001). ``Transfer principles and inequality
aversion, with an application to optimal growth,'' \emph{Mathematical
Social Sciences} 42: 1\textendash 11.

Gaertner, W., \& Schokkaert, E. (2012). Empirical social choice: questionnaire-experimental
studies on distributive justice. Cambridge University Press.

Hurley, J., Mentzakis, E., \& Walli-Attaei, M. (2020). Inequality
aversion in income, health, and income-related health. Journal of
health economics, 70, 102276.

Kolm, S. C. (1976). Unequal inequalities. I. Journal of economic Theory,
12(3), 416-442.

Okun, A. M. (2015). Equality and efficiency: The big tradeoff. Brookings
Institution Press.

Thon D., S.W. Wallace (2004). ``Dalton transfers, inequality and
altruism,'' Social Choice and Welfare 22: 447-465.

Tossavainen T., P. Haukkanen (2007). ``The Schröder Equation and
Some Elementary Functions'', \emph{Int. Journal of Math. Analysis}
1(3): 101\textendash 106.

Venmans, F., \& Groom, B. (2021). Social discounting, inequality aversion,
and the environment. Journal of Environmental Economics and Management,
109, 102479.

\section{Appendix}

\subsection{Proof of Proposition \ref{prop:The-Kolm-Pollak-subclass}\label{subsec:Proof-of-Proposition-5}}
\begin{proof}
From Prop. \ref{prop:constcodam21} we know that $f\left(y\right)=-e^{g\left(y\right)}$
for a function $g$ such that 
\[
g\left(\left(k+1\right)\varDelta+a\right)-g\left(k\varDelta+a\right)=-\ln2
\]
for all $k\in\mathbb{Z}_{+},$$a\in\left[0,\varDelta\right]$. By
a similar reasoning, the function $g$ must also satisfy
\[
g\left(\left(k+1\right)\varOmega+a\right)-g\left(k\varOmega+a\right)=-\ln3
\]
for all $k\in\mathbb{Z}_{+},$$a\in\left[0,\varOmega\right]$. In
particular, for all $k\in\mathbb{Z}_{+},$
\[
g\left(k\varDelta\right)-g\left(0\right)=-k\ln2
\]
\[
g\left(k\varOmega\right)-g\left(0\right)=-k\ln3.
\]
Given the quasi-periodic behavior of $g$, there exist $b<0<c$ such
that for all $x\in\mathbb{R}$, 
\[
b<g\left(x\right)-g\left(0\right)+x\frac{\ln2}{\varDelta}<c
\]
\[
b<g\left(x\right)-g\left(0\right)+x\frac{\ln3}{\varOmega}<c.
\]
This is possible only if 
\[
\frac{\ln2}{\varDelta}=\frac{\ln3}{\varOmega}.
\]
Let $\alpha$ be this common ratio. This proves that $\varDelta=\frac{\ln2}{\alpha}$
and $\varOmega=\frac{\ln3}{\alpha}$ for some $\alpha>0$.

Without loss of generality, we can let $\ensuremath{f}\left(+\infty\right)=0$
and $\ensuremath{f}\left(0\right)=-1$. Then

\[
\left\{ \begin{array}{c}
f\text{\ensuremath{\left(y-\varDelta\right)}}=2\ensuremath{f}\left(y\right)\\
f\text{\ensuremath{\left(y-\varOmega\right)}}=3\ensuremath{f}\left(y\right).
\end{array}\right.
\]

Letting $h=f\circ ln$, a change of variables $y=\ln x$ yields

\[
\left\{ \begin{array}{c}
h\text{\ensuremath{\left(2^{-\frac{1}{\alpha}}x\right)}}=2\ensuremath{h}\left(x\right)\\
h\text{\ensuremath{\left(3^{-\frac{1}{\alpha}}x\right)}}=3\ensuremath{h}\left(x\right).
\end{array}\right.
\]

The rest of the argument follows the proof of the bottom half of Prop.
\ref{prop:fixedprop123}, for $\eta=1+\alpha$, from where we obtain
that $h\text{\ensuremath{\left(x\right)=h\left(1\right)x^{-\alpha}} for all \ensuremath{x\in\left(0,+\infty\right).}}$
Then $f\left(y\right)=h(e^{y})=h\left(e^{0}\right)e^{-\alpha y}=f\left(0\right)e^{-\alpha y}=-e^{-\alpha y}$
for all $y\in\mathbb{R}.$
\end{proof}

\subsection{Proof of Proposition \ref{CPIE}\label{subsec:Proof-of-Proposition-seven}\protect 
}\begin{proof}
The proof mimics previous proofs. First, one shows that $\ddot{\beta}=2^{\ensuremath{\frac{1}{1-\gamma}}},\dddot{\beta}=3^{\ensuremath{\frac{1}{1-\gamma}}}$.
To do this, we note that, by definition, $\frac{1}{2}f\text{\ensuremath{\left(\ddot{Y}\left(y\right)\right)}}+\frac{1}{2}\text{\ensuremath{f}\ensuremath{\left(\infty\right)}}=\ensuremath{f}\left(y\right)$.
Notice that $\ddot{Y}\left(y\right)$$\in\left(0,y\right)$ and since
$f$ is monotone increasing this implies that $f$ is bounded from
above. Without loss of generality, we can let $\ensuremath{f}\left(+\infty\right)=0$
and $\ensuremath{f}\left(c\cdot e\right)=\frac{1}{1-\text{\ensuremath{\gamma}}}$.
Then $f\text{\ensuremath{\left(y^{\ddot{\beta}}c^{1-\ddot{\beta}}\right)}}=2\ensuremath{f}\left(y\right)$.
Similarly, we know that $\frac{1}{3}f\text{\ensuremath{\left(\dddot{Y}\left(y\right)\right)}}+\frac{1}{3}f\text{\ensuremath{\left(\infty\right)}}+\frac{1}{3}\text{\ensuremath{f}\ensuremath{\left(\infty\right)}}=\ensuremath{f}\left(y\right)$
and so $f\text{\ensuremath{\left(y^{\dddot{\beta}}c^{1-\dddot{\beta}}\right)}}=3\ensuremath{f}\left(y\right)$.
Let $z=\ln y-\ln c$ and $h\left(z\right)=f\left(c\cdot e^{z}\right)$.
Then the equations $f\text{\ensuremath{\left(y^{\ddot{\beta}}c^{1-\ddot{\beta}}\right)}}=2\ensuremath{f}\left(y\right)$
and $f\text{\ensuremath{\left(y^{\dddot{\beta}}c^{1-\dddot{\beta}}\right)}}=3\ensuremath{f}\left(y\right)$
respectively become

\[
\left\{ \begin{array}{c}
h\text{\ensuremath{\left(\ddot{\beta}z\right)}}=2\ensuremath{g}\left(z\right)\\
h\text{\ensuremath{\left(\dddot{\beta}z\right)}}=3\ensuremath{g}\left(z\right).
\end{array}\right.
\]
Now let $g=\ln\left(-h\circ e\right)$. A second change of variables
$x=\ln z$ yields

\[
\left\{ \begin{array}{c}
\ln2+g\left(x\right)=g\left(\ln\ddot{\beta}+x\right)\\
\ln3+g\left(x\right)=g\left(\ln\dddot{\beta}+x\right).
\end{array}\right.
\]

We obtain that, for all $k\in\mathbb{Z}_{+},$
\[
g\left(k\ln\ddot{\beta}\right)-g\left(0\right)=k\ln2
\]
\[
g\left(k\ln\dddot{\beta}\right)-g\left(0\right)=k\ln3
\]
which, as in the proof of Prop. \ref{prop:fixedprop123}, is only
possible if 
\[
\frac{\ln2}{\ln\ddot{\beta}}=\frac{\ln3}{\ln\dddot{\beta}}.
\]
We now let $1-\gamma$ be this common (negative) ratio. This proves
that $\ddot{\beta}=2^{\ensuremath{\frac{1}{1-\gamma}}}$ and $\dddot{\beta}=3^{\ensuremath{\frac{1}{1-\gamma}}}$
for some $\gamma>1$. 

The rest of the argument is as in the second half of the proof of
Prop. \ref{prop:fixedprop123}, from where it follows that $\ensuremath{h}\left(z\right)=h\left(1\right)z^{1-\gamma}$
for all $z\in\left(0,\infty\right)$. Since $h\left(1\right)=f\left(c\cdot e\right)=\frac{1}{1-\text{\ensuremath{\gamma}}}$,
we obtain that $f\left(y\right)=g\left(\ln y-\ln c\right)=\frac{\left(\ln\frac{y}{c}\right)^{1-\gamma}}{1-\gamma}$
for all $y\in\left(c,\infty\right)$, which is what we wanted to show.
\end{proof}

\end{document}